\documentclass[12pt,authoryear, review]{elsarticle}
\usepackage{geometry}
\usepackage{amsthm}
\usepackage{amsmath}
\usepackage{amssymb}

\makeatletter
  \theoremstyle{definition}
  
  \theoremstyle{plain}
  \newtheorem{prop}{\protect\propositionname}
  \theoremstyle{plain}
  \newtheorem{cor}{\protect\corollaryname}

\newtheorem{example}{Example}
\newtheorem{remark}{Remark}
\newtheorem{definition}{Definition}

\usepackage{scalefnt}
\usepackage{amsfonts}
\usepackage[singlespacing]{setspace}
\usepackage[bottom]{footmisc}
\usepackage{indentfirst}
\usepackage{endnotes}
\usepackage{graphicx}
\usepackage{rotating}
\usepackage{tikz}
\usepackage{pgf}
\usepackage[colorlinks,citecolor=blue,urlcolor=blue]{hyperref}

\newcommand{\com}[1]{}

\def\@biblabel#1{\hspace*{-\labelsep}}

\renewcommand{\k}{\kappa}

\providecommand{\corollaryname}{Corollary}
\providecommand{\definitionname}{Definition}

\providecommand{\propositionname}{Proposition}

\providecommand{\corollaryname}{Corollary}
\providecommand{\definitionname}{Definition}
\providecommand{\propositionname}{Proposition}

\makeatother

  \providecommand{\definitionname}{Definition}
  \providecommand{\propositionname}{Proposition}
\providecommand{\corollaryname}{Corollary}

\begin{document}
\begin{frontmatter}

\title{Simple Characterizations of  Potential Games and Zero-sum Games}

\author[A1]{Sung-Ha Hwang\corref{cor1}}

\ead{sungha@kaist.ac.kr}

\author[A2]{Luc Rey-Bellet}

\ead{luc@math.umass.edu }

\cortext[cor1]{\today.  Corresponding author. The research of S.-H. H. was supported by the Ministry of Education of the Republic of Korea and the National Research Foundation of Korea (NRF-2016S1A5A8019496). The research of L. R.-B. was supported
by the US National Science Foundation (DMS-1109316). }

\address[A1]{Korea Advanced Institute of Science and Technology (KAIST), Seoul, Korea}

\address[A2]{Department of Mathematics and Statistics, University of Massachusetts
Amherst, MA, U.S.A.}

\begin{abstract}

We provide several tests to determine whether a game is a potential game or whether it is a zero-sum equivalent game---a game which is strategically equivalent to
a zero-sum game in the same way that a potential game is
strategically equivalent to a common interest game. We present a unified framework
applicable for both potential and zero-sum  equivalent games by deriving a simple but useful characterization of these games. \com{\textbf{Here: the emphasis on the unified framework is not important.}}
This allows us to re-derive  known criteria  for potential
games, as well as  obtain several new criteria. In particular, we prove
(1) new integral tests for potential games and for
zero-sum equivalent games, (2) a new derivative test for zero-sum
equivalent games, and (3) a new representation characterization for zero-sum
equivalent games.

 \end{abstract}
\end{frontmatter}{

\thispagestyle{empty}

\newpage{}

\section{Introduction\setcounter{page}{1}}
We provide several tests to determine whether a game
is a potential game or whether it is a zero-sum equivalent game---a game which is \emph{strategically equivalent} to
a zero-sum game in the same way that a potential game is
strategically equivalent to a common interest game (see Definition \ref{def: zr}
and also Section 11.2 in \citet{Hofbauer98}). \com{\textbf{Here we need to explain why people should care about zero-sum equivalent games; possibly relate it to concave games when it satisfies some special properties}. } We present a unified framework
applicable for both potential and zero-sum  equivalent games by deriving a simple but useful characterization of these games.
In particular, we propose (1) new integral tests for potential games and for
zero-sum equivalent games, (2) a new derivative test for zero-sum
equivalent games, and (3) a new representation characterization for zero-sum
equivalent games. We also re-derive  known criteria  for potential
games, such as \citet{Monderer96}, \citet{Ui00} and \citet{Sandholm10},
as well as  obtain several new criteria.  \com{\textbf{Extend the literature}}

An advantage of our approach is that our new integral tests can be applied
to normal form games with continuous strategy sets as well as those
with finite strategy spaces, whether payoff functions are
discontinuous or not. Many popular games with continuous strategy sets,
such as Bertrand competition games and Hotelling games, have discontinuous payoff functions.
It is well-known that games with  continuous strategy sets and discontinuous payoff pose special challenges such as the existence of Nash equilibria (see, for example, \citet{Reny99}). Our integral test provides a useful tool to study this class of games. In the case of finite strategy sets our test
reduces to the test in \citet{Sandholm10}.
%

The integral test for potential games is also easier to implement than the
cycle condition in \citet{Monderer96}'s Theorem 2.8 (see Remark \ref{rem:path:conditions}).
For, say,  a two-player game our integral test requires checking the values of a function
at two different points, while the cycle condition requires checking the
values of  a function at four different points.  For finite strategy sets, \citet{Hino11} and \citet{Sandholm10}'s algorithms checking for potential games
have complexity $O(n^2)$ and the integral test has the same complexity.
%

We also study in detail zero-sum equivalent games and provide integral and derivative
tests as well as representations of those games.  While the derivative test for potential games is well-known
(\citet{Monderer96} Theorem 4.5), the derivative test for zero-sum equivalent games is
new and provides an easy and convenient way to check if a game is zero-sum equivalent
when the payoff function is sufficiently smooth (Proposition \ref{prop:der}).
The usefulness of this test is illustrated in Example 2\cite{} where we analyze contest games.
Finally, we provide a representation characterization (Proposition \ref{prop: repII})
which generalizes to zero-sum equivalent games the result in \citet{Ui00}.

%
%

In the existing literature, conditions for potential games, such as
\citet{Monderer96}, \citet{Ui00} and \citet{Sandholm10}, are regarded
as distinct and  derived by different methods (see, e.g., the discussion
in Section 3 in \citet{Sandholm10}). Our result provides a unified
framework to understand and generalize (also to zero-sum equivalent games)
these conditions.

\section{Main Results}
\com{
\textbf{General comments: need to explain the meaning of propositions before the statement; also, explain why do we care...etc?} \\
\textbf{Need to find more literature}
}
We follow the setup in \citet{HandR2017} where we provide general decomposition theorems for $n-$player games. Let $S= S_1 \times \cdots \times S_n$ be the space of all strategy profiles
where $S_i$ is the set of strategies for the $i^{th}$ player.  Let $m_{i}$ be a finite measure on $S_{i}$
and $m$ be the product measure $m:=m_{1}\times\cdots\times m_{n}$.
We denote by $u^{(i)}$ the payoff function for the $i^{th}$ player, where $u^{(i)}:S\to\mathbb{R}$
is a (measurable) function. For fixed $n$ and $S$, a game is uniquely specified by the vector-valued function $u:=(u^{(1)},$ $u^{(2)},$ $\cdots,u^{(n)})$. We use the notation  $g(s_{-i})$ for a function which does not depend on its  $i$-th argument.  If the payoff for the $i^{th}$
player has the form $u^{(i)}(s) = g^{(i)}(s_{-i})$ then her payoff does not depend on her own strategy (also called a {\em passive}
game). It is easy to see that if two game payoffs differ by a passive game for each player, then they have the same Nash equilibria and best response functions---these are called {\em strategically equivalent}.


%
%

\begin{definition}\label{def: zr} We have:\\
\noindent {(i)}
A game $u$ is a \textbf{potential game} if there exists a function $v$ and functions $g^{(i)}$'s such that
\[
(u^{(1)}(s),u^{(2)}(s),\cdots,u^{(n)}(s))=(v(s),v(s),\cdots,v(s))+(g^{(1)}(s_{-1}),g^{(2)}(s_{-2}),\cdots,g^{(n)}(s_{-n})) \,.
\]
for all $s$. \\
\noindent {(ii)}
A game $u$ is a \textbf{zero-sum equivalent game} if there exists functions $v^{(i)}$'s with $\sum_{i} v^{(i)} = 0$
and functions $g^{(i)}$'s  such that
\[
(u^{(1)}(s),u^{(2)}(s),\cdots,u^{(n)}(s))=(v^{(1)}(s),v^{(2)}(s),\cdots,v^{(n)}(s))+(g^{(1)}(s_{-1}),g^{(2)}(s_{-2}),\cdots,g^{(n)}(s_{-n})) \,.
\]
for all $s$.
\end{definition}
%
%
%
\noindent
The definition of a potential game in \citet{Monderer96} is that
$u$ is a potential game if  there exists a function $v$ such that
$u^{(i)}(s_{i},s_{-i})-u^{(i)}(\widetilde{s}_{i},s_{-i})=v(s_{i},s_{-i})-v(\widetilde{s}_{i},s_{-i})$
for all $s_{i},\widetilde{s}_i,s_{-i}$ and all $i$.  This is easily shown
to be equivalent to Definition \ref{def: zr}.

The next proposition is simple but important since it recasts the definitions of potential and zero-sum
equivalent games without reference to unknown functions $v$ or $v^{(i)}$ in Definition \ref{def: zr}.  This will provide the key ingredient
to establish our criteria.
%
\begin{prop}[\textbf{Characterization}]
\label{prop: char}We have:\\
 \noindent (i) A game $u$ is a potential game if and only if  there exist functions $g^{(i)}$'s such that for all $i,j$
\begin{equation}
u^{(i)}(s)-g^{(i)}(s_{-i})=u^{(j)}(s)-g^{(j)}(s_{-j})\label{eq: new_con_p} \,
\end{equation}
for all $s$.\\
 (ii) A game $u$ is a zero-sum equivalent game if and only if there exist functions $g^{(i)}$'s such that
\begin{equation}
\sum_{i=1}^{n}\left[u^{(i)}(s)-g^{(i)}(s_{-i})\right]=0\label{eq: new_con_z}\,
\end{equation}
for all $s$.
\begin{proof}
The ``only if'' parts are trivial. Conversely, let us assume that there exist $g^{(i)}$'s  which satisfy the
conditions (\ref{eq: new_con_p}) or (\ref{eq: new_con_z}). Then, if we write
\begin{align*}
(u^{(1)}(s),u^{(2)}(s),\cdots,u^{(n)}(s)) & =(u^{(1)}(s)-g^{(1)}(s_{-1}),u^{(2)}(s)-g^{(2)}(s_{-2}),\cdots,u^{(n)}(s)-g^{(n)}(s_{-n}))\\
 & +(g^{(1)}(s_{-1}),g^{(2)}(s_{-2}),\,\cdots,g^{(n)}(s_{-n}))
\end{align*}
we see that $u$ is a potential game if (\ref{eq: new_con_p})
holds and that $u$ is a zero-sum equivalent game if (\ref{eq: new_con_z})
holds.
\end{proof}
\end{prop}

For our integral test, we introduce the following operator.

\begin{definition}
\label{lem-ti} For an integrable function $h$:$S\rightarrow\mathbb{R},$
we define $T_{i}$ by
\[
T_{i}h(s)\,:=\frac{1}{m_{i}(S_{i})}\int h(s)dm_{i}(s_{i}).
\]
\end{definition}

Note that $T_{i}$ and $T_{j}$ commute and that we have the identity
\begin{equation}\label{eq: fac_p}
(I-T_{i})(I-T_{j})=I-(T_{i}+(I-T_{i})T_{j})\,,
\end{equation}
where $I$ is the identity operator.
And, by induction,
\begin{equation}
\prod_{l=1}^{n}(I-T_{l})=I-(T_{1}+\sum_{j=2}^{n}\prod_{l=1}^{j-1}(I-T_{l})T_{j})\label{eq: fac_z} \,.
\end{equation}
Note as well for any $h$,  $T_i h$ does not depend on $s_i$.
We next prove our integral tests.
\begin{prop}[\textbf{Integral Tests}]
\label{prop:int}We have:\\
 (i) A game $u$ is a potential game if and only if
\begin{equation}
(I-T_{i})(I-T_{j})(u^{(i)}-u^{(j)})=0\,.\label{eq: gen_con}
\end{equation}
for all $i,j$.
 \\
\noindent
 (ii) A game $u$ is a zero-sum equivalent game if and only if
\begin{equation}
\sum_{i=1}^{n}\prod_{l=1}^{n}(I-T_{l})u^{(i)}=0\,.\label{eq: gen_con-1}
\end{equation}
\end{prop}
\begin{proof}
Suppose that a game is a potential game (or a zero-sum equivalent game). Equation (\ref{eq: gen_con}) (or  (\ref{eq: gen_con-1})) follows from equation (\ref{eq: new_con_p}) (or equation (\ref{eq: new_con_z})) in Proposition \ref{prop: char}. Thus, only if parts in (i) and (ii) hold. For the if part in (i), suppose that \eqref{eq: gen_con} holds.
Let $i$ be fixed. Note that
  \[
    I = (I-T_i) + T_i = \sum_{\substack{M \subset N \\M \ni i} } \prod_{l \not \in M} \prod_{k \in M} T_l(I-T_k) + T_i
  \]
  Thus for $u$,
  \begin{align*}
      u& = (\sum_{\substack{M \subset N \\M \ni 1} } \prod_{l \not \in M} \prod_{k \in M} T_l(I-T_k) u^{(1)}, \cdots, \sum_{\substack{M \subset N \\M \ni n} } \prod_{l \not \in M} \prod_{k \in M} T_l(I-T_k) u^{(n)}) + (T_1 u^{(1)}, \cdots, T_n u^{(n)})
  \end{align*}
  Let $N$ be the set of all players, $N=\{1, \cdots, n \}$. Let $M \subset N$ and $i,j \in M$ and $|M| \geq 2$. Then we have
  \begin{align} \label{eq:inv}
    \prod_{k\in M}\prod_{k'\not\in M} (I-T_{k}) {T}_{k'}u^{(i)} & =\prod_{\substack{k\in M\\ k\neq i,j}} \prod_{k'\not\in M}(I-T_{k}) {T}_{k'}(I-T_{i})(I-T_{j})u^{(i)} \notag \\
     &     =\prod_{\substack{k\in M\\k\neq i,j}}\prod_{k'\not\in M}(I-T_{k}) {T}_{k'}(I-T_{i})(I-T_{j})u^{(j)}=\prod_{k\in M}\prod_{k'\not\in M}(I-T_{k}) {T}_{k'}u^{(j)}
  \end{align}
  Thus if $|M|>2$, then $\prod_{k\in M}\prod_{k'\not\in M} (I-T_{k}) {T}_{k'}u^{(i)} = \prod_{k\in M}\prod_{k'\not\in M} (I-T_{k}) {T}_{k'}u^{(j)}$ for all $i,j \in M$.  Thus we can define
  \[
    \xi_M := \prod_{k\in M}\prod_{k'\not\in M} (I-T_{k}) {T}_{k}u^{(i)}
  \]
  for any $i \in M \subset N$.
  Thus we have
  \[
    u^{(i)} =\sum_{\substack{M \subset N \\M \ni i} } \prod_{l \not \in M} \prod_{k \in M} T_l(I-T_k)u^{(i)} + T_iu^{(i)} = \sum_{\substack{M \subset N \\ M \ni i}}\xi_M  + T_iu^{(i)} =  \sum_{\substack{M \subset N \\ M \neq \emptyset}}\xi_M - \sum_{\substack{M \not \ni i \\ M \neq \emptyset}}\xi_M  +T_iu^{(i)}
\]
and thus $u$ is a potential game.\\
For the if part in (ii), from (\ref{eq: fac_z}) we obtain
\[
\sum_{i=1}^{n}\prod_{l=1}^{n}(I-T_{l})u^{(i)}=\prod_{l=1}^{n}(I-T_{l})\sum_{i=1}^{n}u^{(i)}=0\text{ if and only if }\sum_{i=1}^{n}u^{(i)}=T_{1}\sum_{i=1}^{n}u^{(i)}+\sum_{j=2}^{n}\prod_{l=1}^{j-1}(I-T_{l})T_{j}\sum_{i=1}^{n}u^{(i)}.
\]
Again, observe that $T_{1} \sum_{i=1}^{n}u^{(i)}$ does not depend on $s_{1}$ and $\prod_{l=1}^{j-1}(I-T_{l})T_{j}\sum_{i=1}^{n}u^{(i)}$ does not depend on $s_{j}.$ Thus from Proposition \ref{prop: char}, $u$ is a zero-sum equivalent game.
\end{proof}

\begin{remark}{\bf (The cycle condition)}\label{rem:path:conditions}
{\rm
The integral test can be compared to the well-known cycle condition of \citet{Monderer96} (Theorem 2.8).  Consider a two player game. The cycle condition requires the following four variable function, $\Phi(s_{1},s_{2}, \widetilde{s_{1}}, \widetilde{s_{2}}  )$, to be identically zero
\[
\Phi(s_{1},s_{2}, \widetilde{s_{1}}, \widetilde{s_{2}}  ):= \left[ u^{(1)}(\widetilde{s}_{1}, {s}_{2}) - u^{(1)}(s_{1}, s_{2})\right]
+ \left[ u^{(2)}(\widetilde{s}_{1}, \widetilde{s}_{2}) - u^{(2)}(\widetilde{s}_{1}, s_{2})\right]
\]
\[
+  \left[ u^{(1)}({s}_{1}, \widetilde{s}_{2}) - u^{(1)}(\widetilde{s}_{1}, \widetilde{s}_{2})\right]
+ \left[ u^{(2)}({s}_{1}, {s}_{2}) - u^{(2)}({s}_{1}, \widetilde{s}_{2})\right],
\]
while our integral test requires that the following two variable function, $\Psi(s_1,s_2)$, to be identically zero
\[
 \Psi(s_1, s_2):= u^{(1)}(s_{1},s_{2})-\frac{1}{|S_{1}|}\int u^{(1)}(s_{1}',s_{2})ds_{1}'-\frac{1}{|S_{2}|}\int u^{(1)}(s_{1},s_{2}')ds_{2}'+\frac{1}{|S_{1}||S_{2}|}\int u^{(1)}(s_{1}',s_{2}')ds_{1}'ds_{2}'
\]
\[
-u^{(2)}(s_{1},s_{2})+\frac{1}{|S_{1}|}\int u^{(2)}(s_{1}',s_{2})ds_{1}'+\frac{1}{|S_{2}|}\int u^{(2)}(s_{1},s_{2}')ds_{2}'-\frac{1}{|S_{1}||S_{2}|}\int u^{(2)}(s_{1}',s_{2}')ds_{1}'ds_{2}'
\]
Thus, the cycle condition requires checking the values of a function of four variables, while
the integral test for potential games requires checking the values of a function of two variables---this implies a significant reduction
of the computational complexity. For instance, if we numerically compare two functions at $n$ different points, the number of equalities to be checked under our test is order $n^2$, while this number under the cycle condition test becomes order $n^4$ (see the related discussion on p.200 in \citet{Hino11}).  Note as well that in \citet{HandR2017}  we prove a cycle-like condition for games which  are zero-sum equivalent.


}
\end{remark}

%
%
%

If $S$ is a finite set and $m$ is the counting
measure, then the integral test for potential games becomes
the condition by \citet{Sandholm10}. For the convenience of the reader,
we provide a two-player version.
\begin{cor}[\citealt{Sandholm10}]
\label{cor: 2p-int-1} A two player game with payoff matrices $(A,B)$ is a potential game if and only if
\begin{equation*} A_{ij}-\frac{1}{|S_{1}|}\sum_{i}A_{ij}-\frac{1}{|S_{2}|}\sum_{j}A_{ij}+\frac{1}{|S_{1}|}\frac{1}{|S_{2}|}\sum_{i,j}A_{ij}
=  B_{ij}-\frac{1}{|S_{1}|}\sum_{i}B_{ij}-\frac{1}{|S_{2}|}\sum_{j}B_{ij}+\frac{1}{|S_{1}|}\frac{1}{|S_{2}|}\sum_{i,j}B_{ij}  \,.
\end{equation*}
\end{cor}

For the derivative test one needs to assume that strategy sets $S_i$ consist of intervals and
that payoff functions  $u^{(i)}$ are twice continuously differentiable on $S$.  An elementary fact from calculus
is that if function $g$ is twice continuously differentiable, then
\[
\frac{\partial^{2}g}{\partial s_{i}\partial s_{j}}(s)=0\,\,\text{if\,\ and only if}\,\, g(s)=G(s_{-i})+K(s_{-j})
\]
for some $G$ and $K$. From this, it is easy to derive
a derivative test for potential games (\citet{Monderer96}, Theorem
4.5). We also provide a similar test for zero-sum equivalent games.
\begin{prop}[\textbf{Derivative Tests}]
\label{prop:der}Assume that the strategy sets are intervals. Then we have:\\
 (i)(\citealt{Monderer96}) If $u$ is twice-continuously
differentiable, the game $u$ is a potential game if and only if for all $i,j$
\begin{equation}
\frac{\partial^{2}u^{(i)}}{\partial s_{i}\partial s_{j}}(s)=\frac{\partial^{2}u^{(j)}}{\partial s_{i}\partial s_{j}}(s).\label{eq: gen_con-2}
\end{equation}
for all $s$.
\\
 (ii) If $u$ is n-times continuously differentiable, the
game $u$ is zero-sum equivalent  if and only if
\begin{equation}
\sum_{i=1}^{n}\frac{\partial^{n}u^{(i)}}{\partial s_{1}\partial s_{2}\cdots\partial s_{n}}(s)=0.\label{eq: gen_con-1-1}
\end{equation}
for all $s$.
\end{prop}
\begin{proof}
Again, from Proposition \ref{prop: char} ``only parts'' easily follow.  For ``if'' parts, (i)
follows from the remark before Proposition \ref{prop:der}. For (ii), we observe
that
\[
\frac{\partial^{n}\sum_{i=1}^{n}u^{(i)}}{\partial s_{1}\partial s_{2}\cdots\partial s_{n}}(s)=0\text{\,\ ~if and only if~}\,\sum_{i=1}^{n}u^{(i)}(s)=g^{(1)}(s_{-1})+g^{(2)}(s_{-2})+\cdots+g^{(n)}(s_{-n}).
\]
\end{proof}

Finally, our last results are alternative representations which are useful to identify games.
\begin{prop}[\textbf{Representation}]
\label{prop: repII}We have:\\
 (i) (\citealt{Ui00}) A game $u$ is a potential game if and only
if there exist functions $w$ and $h^{(i)}$'s such that
\begin{equation}
u^{(i)}(s)=w(s)+\sum_{l\neq i}h^{(l)}(s_{-l}).\label{eq: pot-1}
\end{equation}
for all $s$. \\
 (ii) A game $u$ is a zero-sum equivalent game if and only if there exist a constant $c$,
functions $w^{(i)}$'s and $h^{(i)}$'s such that $\sum_{i}w^{(i)}(s)=c$ and
\[
u^{(i)}(s)=w^{(i)}(s)+\sum_{l\neq i}h^{(l)}(s_{-l}) \,.
\]
for all $s$.
\end{prop}
\begin{proof}
Observe that for the ``if'' part in (i)
\[
u^{(i)}(s)-u^{(j)}(s)=\sum_{l\neq i}h^{(l)}(s_{-l})-\sum_{l\neq j}h^{(l)}(s_{-l})=h^{(j)}(s_{-j})-h^{(i)}(s_{-i}).
\]
If we let $g^{(i)}(s_{-i}):=-h^{(i)}(s_{-i})$, then the asserted claim follows from Proposition \ref{prop: char}.
For the ``if'' part in (ii), we have
\[
\sum_{i=1}^{n}u^{(i)}(s)=c+\sum_{i=1}^{n}\sum_{l\neq i} h^{(l)}(s_{-l}) =c + \sum_{l=1}^n \sum_{i \neq l} h^{(l)}(s_{-l})=\sum_{l=1}^n (\frac{c}{n} + (n-1) h^{(l)}(s_{-l})),
\]
and if we let $g^{(l)}(s_{-l}):=\frac{c}{n} + (n-1) h^{(l)}(s_{-l})$, then the assertion follows from Proposition \ref{prop: char}. Conversely, let  $u$ be a potential game. Then from Proposition \ref{prop: char}, there exist function $g^{(i)}$'s satisfying \eqref{eq: new_con_p}.  Then we write
\[
u^{(i)}(s)=\underbrace{u^{(i)}(s)-g^{(i)}(s_{-i})+g^{(i)}(s_{-i})+\sum_{l\neq i}g^{(l)}(s_{-l})}_{=:w(s)}-\sum_{l\ne i}g^{(l)} (s_{-l}).
\]
and $h^{(l)}(s_{-l}):=-g^{(l)}(s_{-l})$.
Similarly, if $u$ is a zero-sum equivalent, from Proposition  \ref{prop: char}, there exist function $g^{(i)}$'s satisfying \eqref{eq: new_con_z}. Then,
\[
u^{(i)}(s)=\underbrace{u^{(i)}(s)-g^{(i)}(s_{-i})+g^{(i)}(s_{-i})-\frac{1}{n-1}\sum_{l\neq i}g^{(l)}(s_{-l})}_{=:w^{(i)}(s)}+\frac{1}{n-1}\sum_{l\neq i}g^{(l)}(s_{-l}).
\]
Observe that $\sum_{i=1}^n \frac{1}{n-1} \sum_{l \neq i} g^{(l)}(s_{-l}) = \sum_{l=1}^n g^{(l)}(s_{-l})$. We also have $h^{(l)}(s_{-l})=\frac{1}{n-1}g^{(l)}(s_{-l})$. From these  ``only if'' parts follow.
\end{proof}
The first part of  Proposition \ref{prop: repII} is closely
related  to Theorem 3 in \citet{Ui00}. It is identical for two-player games and
easily seen to be equivalent in general.  Proposition \ref{prop: repII} provides a useful tool
to verify if a game is a potential game or a zero-sum equivalent. For example, if $u^{(i)}(s)=w^{(i)}(s)+h^{(i)}(s_{i})$, as is often the case
in economics models with quasi-linear utility functions where benefit and cost functions are separable,
one ignores the cost term depending on his own strategy to determine if the game is potential or zero-sum equivalent.

\begin{figure}
\centering

\begin{tikzpicture}[xscale=0.5, yscale=0.5]
\draw[dashed] (1,0) -- (9,0);
\draw[dashed] (10,1)--(10,9);
\draw[dashed] (9,10) -- (1,10);
\draw[dashed] (0,9) -- (0,1);
\draw[<->] (0.5,9) -- (4.5,5.5);
\draw[<->] (9.5,9) -- (5.5,5.5);
\draw[<->] (0.5,1) -- (4.5,4.5);
\draw[<->] (9.5,1) -- (5.5,4.5);

\node at (5,5) {Proposition \ref{prop: char}};

\node at (-2,0.5) {Proposition \ref{prop:der}};
\node at (-2,-0.5) {\textbf{Derivative Tests}};

\node at (12.3,0.5) {Proposition \ref{prop:int}};
\node at (12,-0.5) {\textbf{Integral Tests}};

\node at (-2, 11){Remark \ref{rem:path:conditions}};
\node at (-2,10) {\textbf{Cycle Conditions}};

\node at (12.2, 11){Proposition \ref{prop: repII}};
\node at (12.2, 10) {\textbf{Representations}};

\end{tikzpicture}

\protect\caption{\textbf{Relationships between various conditions. } \label{fig:summ} }
\end{figure}

Figure \ref{fig:summ} summarizes the relationships between various conditions developed in this paper. All our conditions are derived from Proposition \ref{prop: char}. We first derive the
integral tests from Proposition \ref{prop: char} (Proposition \ref{prop:int}).
We then derive the derivative tests (Proposition \ref{prop:der})
and derive the representation characterizations (Proposition \ref{prop: repII}).
A cycle condition for zero-sum equivalent games appears in \citet{HandR2017}.

\section{Examples}  We illustrate our results with two\com{ (\textbf{Not two, many!}) } simple examples. First we discuss the integral test for potential games.

\begin{example}{\bf (Integral test for potential games) }{\rm
Consider a two-player game  where the strategy sets are two intervals $S_1$ and $S_2$ with Lebesgue measures $|S_{1}|$ and $|
S_{2}|$, respectively, and the payoffs are $u^{(1)}(s_{1},s_{2})$ and $u^{(2)}(s_{1},s_{2})$.  By definition the game is a potential game  if
the payoff has the form $u^{(1)}(s_{1},s_{2})= v(s_{1},s_{2})+ g(s_2)$ and $u^{(2)}(s_{1},s_{2})= v(s_{1},s_{2}) + h(s_{1})$.  Then it
is easy to check that we have the equality
%
%
\begin{align}
 & u^{(1)}(s_{1},s_{2})-\frac{1}{|S_{1}|}\int u^{(1)}(s_{1},s_{2})ds_{1}-\frac{1}{|S_{2}|}\int u^{(1)}(s_{1},s_{2})ds_{2}+\frac{1}{|S_{1}||S_{2}|}\int u^{(1)}(s_{1},s_{2})ds_{1}ds_{2}\nonumber \\
= & u^{(2)}(s_{1},s_{2})-\frac{1}{|S_{1}|}\int u^{(2)}(s_{1},s_{2})ds_{1}-\frac{1}{|S_{2}|}\int u^{(2)}(s_{1},s_{2})ds_{2}+\frac{1}{|S_{1}||S_{2}|}\int u^{(2)}(s_{1},s_{2})ds_{1}ds_{2}\label{eq: test-1}\,.
\end{align}
Our integral test asserts that  {\em if equation (\ref{eq: test-1}) holds,  the game is actually a potential game}.
%
By the symmetry of the formula in $s_1$ and $s_2$, one also sees that if the payoffs have the form
$u^{(1)}(s_{1},s_{2}):=v(s_{1},s_{2})+g(s_{1})$ and $u^{(2)}(s_{1},s_{2}):=v(s_{1},s_{2})+h(s_{2})$,
then the condition (\ref{eq: test-1}) holds and thus the game is a potential game.
More explicitly, we can write
\[
    (u^{(1)}(s_1, s_2), u^{(2)}(s_1, s_2)) = (v(s_1, s_2)+g(s_1)+h(s_2), v(s_1, s_2)+h(s_2)+g(s_1)) - (h(s_2), g(s_1))
\]
which shows that $u$ is a potential game.  This provides the characterization
of potential games in Proposition \ref{prop: repII}. Although somewhat trivial, this example illustrates our integral test in the simplest possible setting.
}
\end{example}

Next we use our derivative test for a class of contest games.
\com{
\begin{example}{\bf (All pay-auction) }
	Introduce $\Phi$ function. Show uniqueness.
\end{example}
}
\begin{example}{\bf (Contest games) }
{\rm
 Suppose that $S_{1}=S_{2}=(0,\infty)$ and consider the following
contest game (see, e.g., \citet{Konrad07}).  For $f$ positive, define
\begin{equation}
u^{(1)}(s_{1},s_{2})=\frac{f(s_{1})}{f(s_{1})+f(s_{2})}v-c_{1}(s_{1})\textrm{, }\, u^{(2)}(s_{1},s_{2})=\frac{f(s_{2})}{f(s_{1})+f(s_{2})}v-c_{2}(s_{2}) \,.\label{eq:contest}
\end{equation}
We set $p^{(1)}(s_{1},s_{2}):=f(s_{1})/(f(s_{1})+f(s_{2}))$
and $p^{(2)}(s_{1},s_{2}):=1-p^{(1)}(s_{1},s_{2})$ which are  the probabilities of winning a prize of value $v$.
Here, $s_{i}$ is the amount of resources invested in the contest to obtain the  prize  while $c_{i}(s_{i})$ is its associated cost.

Our derivative test for zero-sum equivalent games (see Proposition \ref{prop:der}) asserts that when the payoffs are differentiable, a game is equivalent to a zero-sum game if we have the equality
\[
\frac{\partial^2 u^{(1)}}{\partial s_{1}\partial s_{2}}(s_1, s_2) +\frac{\partial^2 u^{(2)}}{\partial s_{1}\partial s_{2}} (s_1,s_2)=0 \,.
\]
Indeed we have
%
\begin{align*}
\frac{\partial^{2}u^{(1)}}{\partial s_{1}\partial s_{2}}+\frac{\partial^{2}u^{(2)}}{\partial s_{1}\partial s_{2}}
= v \frac{\partial^2 p^{(1)}}{\partial s_1 \partial s_2} + v \frac{\partial^2 p^{(2)}}{\partial s_1 \partial s_2}
=0
\end{align*}
from $p^{(1)}(s_1, s_2) + p^{(2)}(s_1, s_2) = 1$. If $f(s_i)={s_i}^{\alpha}$ where $\alpha\le1$ and $c_i(s_{i})=s_{i},$ the
game in (\ref{eq:contest}) admits a pure strategy Nash equilibrium \citep{Konrad07}. }
\end{example}

\section{Discussion}

We developed systematic ways of studying potential games
and zero-sum equivalent games. We provided  simple characterizations for potential games and zero-sum equivalent games (Proposition \ref{prop: char}), from which
we obtained new integral tests (Proposition \ref{prop:int}), and a  new derivative test for zero-sum equivalent games (Proposition \ref{prop:der}).
Our methods are general and require few assumptions on the game structure; for example,
discontinuous payoff function games can be studied by the integral tests.

The advantage of the integral tests lies in that it can be applied to games with discontinuous payoff functions, as we mentioned earlier. Discontinuous payoff functions are often used in modeling competitive activities such as auctions and contests. The disadvantage of the integral tests is that it is sometimes  difficult to evaluate integrals, and hence implementing the tests may be harder. Since differentiation is easier than integration in general, the derivative tests have an advantage in that it can be implemented easier, with the disadvantage in limited applicability; that is, the derivative tests can only be applied to games with differentiable payoff functions.

\com{

\section{Conclusion}

TBA
}

%

\newpage{}

\singlespace

\scalefont{1}

 \bibliographystyle{elsarticle-harv}
\bibliography{evolutionary_games}

\end{document}